\documentclass[a4paper,UKenglish,cleveref, autoref]{lipics-v2019}

\usepackage{amsmath,stmaryrd,graphicx}

\makeatletter
\newcommand{\fixed@sra}{$\vrule height 2\fontdimen22\textfont2 width 0pt{\mathbf\shortrightarrow}$}
\newcommand{\shortarrow}[1]{%
  \mathrel{\text{\rotatebox[origin=c]{\numexpr#1*45}{\fixed@sra}}}
}
\makeatother

\title{An $\Omega(n^2)$ Lower Bound for Random Universal Sets for Planar Graphs} 

\titlerunning{Random Universal Sets Lower Bound}

\author{Alexander Choi}{Department of Computer Science, 
						University of California, 
						Riverside}
						{achoi035@ucr.edu}{[orcid]}{Research supported by NSF grant CCF-1536026.}
						
\author{Marek Chrobak}{Department of Computer Science, 
						University of California, 
						Riverside}
						{marek@cs.ucr.edu}{[orcid]}{Research supported by NSF grant CCF-1536026.}

\author{Kevin Costello}{Department of Mathematics, 
						University of California, 
						Riverside}
						{costello@math.ucr.edu}{[orcid]}{}

\authorrunning{A. Choi, M. Chrobak, and K. Costello}

\Copyright{Alexander Choi, Marek Chrobak, and Kevin Costello}



\category{}

\relatedversion{}

\supplement{}

\nolinenumbers 

\newcommand{\mareksmargincomment}[1]%
    {{%
      \marginpar{{\tiny\begin{minipage}{0.5in}
                       \begin{flushleft}
                          {\color{red}MCh} {#1}
                       \end{flushleft}
                       \end{minipage}
                }}
    }}

\newcommand{\ignore}[1]{}

\newcommand{\myparagraph}[1]{{\smallskip\noindent{\bf #1}}}

\newcommand{\etal}{{\em et~al.}}

\newcommand{\calL}{{\cal L}}

\newcommand{\braced}[1]{{ \left\{ #1 \right\} }}

\newcommand{\floor}[1]{{ \lfloor #1 \rfloor }}

\newcommand{\Prob}{{\mathbb{P}}}

\newcommand{\reals}{{\mathbb R}}

\newcommand{\permut}{{\textit{perm}}}

\begin{document}

\maketitle

\begin{abstract}
A set $U\subseteq \reals^2$ is $n$-universal if all $n$-vertex planar graphs 
have a planar straight-line embedding into $U$. We prove 
that if $Q \subseteq \reals^2$ consists of
points chosen randomly and uniformly from the unit square then $Q$ must have
cardinality  $\Omega(n^2)$ in order to be $n$-universal with high probability. 
This shows that the probabilistic method, at least in its basic form,
cannot be used to establish an $o(n^2)$ upper bound on universal sets.
\end{abstract}


\section{Introduction}
\label{sec: introduction}




\myparagraph{Planar universal sets.}
Let $n\ge 4$. A set $U$ of points in $\reals^2$ is called \emph{$n$-universal} if each planar graph $G$ with $n$ vertices 
has an embedding into $\reals^2$ that maps one-to-one the vertices of $G$ into points of $U$ and 
maps each edge of $G$ into a
straight-line segment connecting its endpoints, in such a way that these segments do not intersect (except of course
at the endpoints). It is easy to see that for $U$ to be universal, it is sufficient that
the required embedding exists only for \emph{maximal} planar graphs, namely $n$-vertex planar graphs
with $3n-6$ edges. Maximal planar graphs
have the property that in any planar embedding all their faces are $3$-cycles
(for this reason, they are also sometimes called \emph{triangulated}) and that adding any edge
to such a graph destroys the planarity property.


\myparagraph{Past work.}
The main goal of research on universal sets is to construct such sets of small cardinality.
That finite universal sets exist is trivial as, according to F{\'a}ry's theorem~\cite{fary_on-straight-line_1948},
each planar graph has a straight-line embedding in $\reals^2$, so
we can simply embed each $n$-vertex planar graph into a separate set of points. 
For each $n=4,5$ there is only one (non-isomorphic) maximal planar graph, so $n$ points
are sufficient. This was extended by
Cardinal~{\etal}~\cite{cardinal_etal_on-universal_2013}, who showed the existence of $n$-universal sets
of cardinality $n$ for $n \le 10$, as well as their non-existence for $n\ge 15$. 
As shown recently by Scheucher~{\etal}~\cite{scheucher_etal_a-note-on-universal_2018}
with a computer-assisted proof, for $n=11$ at least $12$ points are necessary in a $11$-universal set. 
This leaves the question of
existence of $n$-universal sets of size $n$ open only for $n = 12, 13, 14$.                                                

For arbitrary values of $n$,
various algorithmic upper bounds for $n$-universal sets of size $O(n^2)$ have been described that make use of points on an integer lattice~\cite{bannister_etal_superpatterns_2014,brandenburg_drawing-planar-graphs_2008,chrobak_payne_a-linear-time-algorithm_1995,defraysseix_pach_pollack_how-to-draw_1990,schnyder_embedding-planar-graphs_1990}. The best current upper bound of $\frac{1}{4}n^2 + O(n)$
was given by Bannister~{\etal}~\cite{bannister_etal_superpatterns_2014}.
The technique in~\cite{bannister_etal_superpatterns_2014} involves reducing the problem to a combinatorial question about superpatterns of integer permutations.  

Very little is known about lower bounds for universal sets. Following an earlier sequence of papers~\cite{defraysseix_pach_pollack_how-to-draw_1990,chrobak_karloff_a-lower-bound_1989,kurowski_a-1.234-louwer-bound_2004},
recently Scheucher~{\etal}~\cite{scheucher_etal_a-note-on-universal_2018} proved that $(1.293-o(1))n$
points are required for a set to be universal. 

There is also some research on constructing universal sets for some sub-classes of planar graphs. For
example, Bannister \etal~\cite{bannister_etal_superpatterns_2014}
describe a tight asymptotic bound of $\Theta(n\log n)$ 
for the size of $n$-universal sets for a specific type of ``simply-nested'' graphs.

Summarizing, in spite of 30 years of research, the gap between the lower and upper bounds
for the size of $n$-universal sets is still very large, between linear and quadratic in $n$.  


\myparagraph{Our contribution.} 
The probabilistic method is a powerful tool for proving existence of various combinatorial structures with desired properties. 
In its standard form, it works by establishing a probability distribution on these structures and showing that they have non-zero
probability (typically, in fact, large) of having the required property.
We address here the question whether this approach can work for showing existence of 
smaller $n$-universal sets. Our result is, unfortunately, negative; namely we give an
$\Omega(n^2)$ lower bound for universal sets constructed in this way, as summarized by the following theorem. 


\begin{theorem}\label{thm: main theorem}
Let  $Q \subseteq \reals^2$ be a set of $m$ random points chosen uniformly from the unit square.
If $m \le \left(\frac{n}{48 e}\right)^2$ then with probability at least $1- 8 {\cdot} 4^{-n/12}$ set
$Q$ is not $n$-universal. 
\end{theorem}

The constants in Theorem~\ref{thm: main theorem} are not optimized; they 
were chosen with the simplicity of the proof in mind.

The idea behind the proof of Theorem~\ref{thm: main theorem}  is to reduce the problem to the 
well-studied problem of estimating longest increasing sub-sequences in random permutations.
We present this proof in the next section. 

We stress that Theorem~\ref{thm: main theorem} applies only to the probabilistic method in its
basic form. It does not preclude the possibility that, for some $m = o(n^2)$,
there exist some other probability distribution on the unit square for which the probability of generating a 
universal set is large, or even that the uniform distribution has non-zero probability of 
generating such a set.


\section{Lower Bound Proof}
\label{sec: lower bound proof}



We now prove  Theorem~\ref{thm: main theorem}, starting with a
lemma showing that a universal set must contain
a large monotone subset, as defined below.

Let $U = \braced{ (x_1,y_1), (x_2,y_2), ...,(x_m,y_m) } \subseteq \reals^2$ be a set of
$m$ points in the plane,
where $x_1 < x_2 < ... < x_m$ and $y_i\neq y_j$ for all $i\neq j$.
Consider a subset  $S$ of $U$, say
$S = \braced{ ( x_{j_1} , y_{j_1} ), (x_{j_2} , y_{j_2} ), ..., ( x_{j_\ell} , y_{j_\ell} )}$,
with ${j_1} < {j_2} < ... < {j_\ell}$.
We say that $S$ is \emph{increasing} iff
$y_{j_1} < y_{j_2} < ... < y_{j_\ell}$ and we say that it's  \emph{decreasing}
iff   $y_{j_1} > y_{j_2} > ... > y_{j_\ell}$.
If $S$ is either increasing or decreasing, we call it \emph{monotone}.


\begin{lemma}\label{lem: monotone subsets}
If $U$ is an $n$-universal set then $U$ has a monotone subset $S$ of
cardinality $\floor{n/12}$.
\end{lemma}

\begin{proof}
Without loss of generality we can assume that $n$ is a multiple of $12$, for
otherwise we can apply the argument below to $n' = n - (n\bmod 12)$.

Let $k = n/6$, and let $G_0$ be a maximal planar graph that consists of a sequence of
$2k = n/3$ 3-cycles $C_1, C_2, ..., C_{2k}$
with each consecutive pair of 3-cycles connected by $6$ edges in such a way that for $1 \leq j \leq 2k-1$ the graph between $C_k$ and $C_{k+1}$ is $2$-regular.  This graph is $3$-connected (removing two vertices can only destroy $4$ edges between consecutive layers), so it follows from a result of Whitney \cite{Whitney} that $G_0$ has a topologically unique embedding in the plane up to choice of the external face.  

Since $U$ is universal, $G_0$ has a planar straight-line embedding into $U$.
In this embedding, no matter what face of $G_0$ is selected as the external face,
in the sequence $C_1, C_2, ..., C_{2k}$, either the first $k$ or the last $k$
will be embedded into nested triangles in the plane.
Denote these nested triangles by $T_1,T_2,...,T_k$, listed in order of decreasing area.
In other words, each $T_i$ is inside $T_{i-1}$, for $i > 1$. By definition,
the corner points of each $T_i$ are in $U$.
For each $i$, let $B_i$ denote the bounding box (an axis-parallel rectangle) of $T_i$.
By straighforward geometry,
these bounding boxes are also nested, that is each  $B_i$ is inside $B_{i-1}$, for $i > 1$.
(See Figure~\ref{fig: triangles and boxes}.)

\begin{figure}[ht]
\begin{center}
\includegraphics[width = 3.75in]{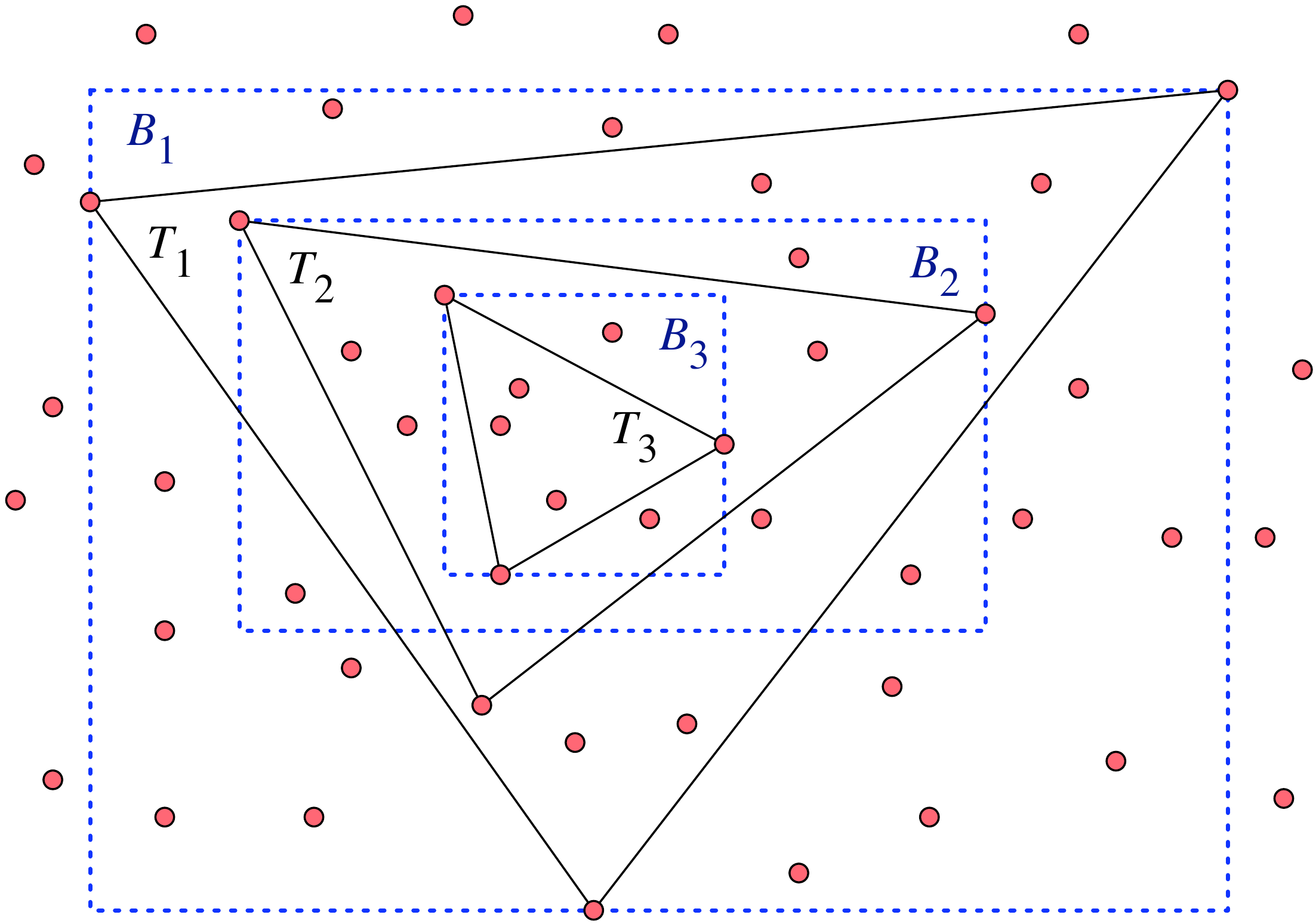}
\caption{Triangles $T_1$, $T_2$, $T_3$ and their bounding boxes.}
\label{fig: triangles and boxes}
\end{center}
\end{figure}

We claim that,
for each $i$, at least one corner of $B_i$ is also a corner of $T_i$, and thus also belongs to $U$.
This is straighforward: each of the four (axis-parallel) sides of $B_i$ must touch one of the three corners of $T_i$.
So there must be a corner of $T_i$ that is touched by two sides of $B_i$. This corner of $T_i$ is then also
a corner of $B_i$, proving our claim.

Let $S'$ be the set of the corners of the bounding boxes $B_i$ that are in $U$. By the claim in
the above paragraph, we have $|S'|\ge k$.  Partition $S'$ into two subsets $S_1$ and $S_2$,
where $S_1$ is the set of points in $S'$ that are either
bottom-left corners or top-right corners of the bounding boxes,
and $S_2$ is the set of points in $S'$ that are either top-left corners or bottom-right corners.
Since the bounding boxes $B_i$ are nested, we obtain that $S_1$
is increasing and $S_2$ is decreasing.

To complete the proof, take $S$ to be either $S_1$ or $S_2$, whichever set is larger,
breaking the tie arbitrarily. Then $S$ is monotone. Furthermore,
since $|S'| \ge k$, $S_1\cup S_2 = S'$, and $S_1$ and $S_2$ are disjoint, we
conclude that $|S|\ge k/2 = n/12$, completing the proof.
\end{proof}


Let $U$ and $S$ be as defined before the statement of Lemma~\ref{lem: monotone subsets}.
With $U$ we can associate a permutation $\pi$ of $\braced{1,2,...,m}$ determined by
having $\pi(i)$ be the rank of $y_i$ in the set of all y-coordinates of $U$, that is
\begin{equation*}
y_{\pi^{-1}(1)} < y_{\pi^{-1}(2)} < ... < y_{\pi^{-1}(m)}.
\end{equation*}
We denote this permutation $\pi$ by $\permut(U)$. Then $S$ naturally \emph{induces} a
subsequence $\pi(j_1)\pi(j_2) ... \pi(j_\ell)$ of $\permut(U)$.
Obviously, $S$ is increasing (resp. decreasing or monotone)
iff its induced subsequence of $\permut(U)$ is increasing (resp.  decreasing or monotone).

\smallskip
We are now ready to prove Theorem~\ref{thm: main theorem}.
Suppose now that $Q$ is a random set of $m$ points chosen uniformly from the unit
square. The probability that any two points have an equal coordinate is $0$, and
thus this event can be simply neglected.
Further, this distribution on sets $Q$ induces a uniform distribution on
the associated permutations $\permut(Q)$ of $\braced{1,2,...,m}$.
Therefore, by Lemma~\ref{lem: monotone subsets}, to prove Theorem~\ref{thm: main theorem}
it is sufficient to show the following claim:


\begin{claim}\label{cla: monotone permutations}
If $\pi$ is a random permutation of $\braced{1,2,...,m}$
and $m\le (\frac{n}{48 e})^2$, then the probability
that $\pi$ contains a monotone subsequence of length $\floor{n/12}$ is at most $8 \cdot 4^{-n/12}$.
\end{claim}

\begin{proof}
The proof of Claim~\ref{cla: monotone permutations} can be derived from the
the argument given by Alan Frieze in~\cite{frieze_on-the-length_1991}. We include
it here for the sake of completeness.

Without loss of generality, we can assume that $n \ge 24$. Otherwise, $m=0$ and the claim is trivially true.

Let $\ell = \floor{\frac{n}{12}} \ge 2e\sqrt{m}$, where the inequality follows from the bound on $m$.
Let $\calL$ be the family of $\ell$-element subsets of $\braced{1,2,...,m}$.
If $L\in \calL$, say $L = \braced{a_1,a_2,...,a_\ell}$ where $a_1 < a_2 < ... < a_\ell$, we will say that
$L$ is \emph{monotone} in a permutation $\pi$ if $\pi(a_1) \pi(a_2) \ldots \pi(a_\ell)$ defines a monotone sequence.
Each monotone subsequence of length $\ell$ can be either increasing or decreasing, so
each $L\in\calL$ has probability ${2}/{\ell!}$ of being monotone.
Using these observations, the union bound, and the inequality
$\ell! \ge (\ell/e)^\ell$ (that follows from Stirling's formula), we obtain:
\begin{align*}
    \textstyle\Prob (\,\exists \;\textrm{monotone}\; L\in\calL \,)
	\;&\le\; \textstyle\sum_{L \in\calL }\Prob (\,L \;\textrm{is monotone}\,) \\
    \;&=\; \binom{m}{\ell}\cdot\frac{2}{\ell!} \\
    \;&=\; 2\cdot \frac {m \cdot (m-1)\cdot \ldots \cdot (m-\ell+1)} {(\ell!)^2} \\
    \;&\le\; 2\cdot \frac{m^{\ell}}{(\ell/e)^{2\ell}} \\
    \;&=\; 2 \cdot \left(\frac{me^2}{\ell^2}\right)^{\ell} \\
    \;&\le\; 2 \cdot \left(\frac{me^2}{\left(2e\sqrt{m}\right)^2}\right)^{\ell} \\
    \;&=\; \textstyle2\cdot 4^{-\ell} \\
    \;&\le\; \textstyle8 \cdot 4^{-n/12},
\end{align*}
since $\ell\ge n/12-1$. This completes the proof of Claim~\ref{cla: monotone permutations} and
Theorem~\ref{thm: main theorem}.
\end{proof}


\bibliography{references_universal_sets}

\begin{thebibliography}{10}

\bibitem{bannister_etal_superpatterns_2014}
{Michael}~J. {Bannister}, {Zhanpeng} {Cheng}, {William}~E. {Devanny}, and
  {David} {Eppstein}.
\newblock Superpatterns and universal point sets.
\newblock {\em Journal of Graph Algorithms and Applications}, 18:177--209,
  2014.
\newblock \href {http://dx.doi.org/10.7155/jgaa.00318}
  {\path{doi:10.7155/jgaa.00318}}.

\bibitem{brandenburg_drawing-planar-graphs_2008}
Franz~J. Brandenburg.
\newblock Drawing planar graphs on $\frac{8}{9}n^2$ area.
\newblock {\em Electronic Notes in Discrete Mathematics}, 31:37 -- 40, 2008.
\newblock The International Conference on Topological and Geometric Graph
  Theory.
\newblock URL:
  \url{http://www.sciencedirect.com/science/article/pii/S1571065308000619},
  \href {http://dx.doi.org/https://doi.org/10.1016/j.endm.2008.06.005}
  {\path{doi:https://doi.org/10.1016/j.endm.2008.06.005}}.

\bibitem{cardinal_etal_on-universal_2013}
Jean Cardinal, Michael Hoffmann, and Vincent Kusters.
\newblock On universal point sets for planar graphs.
\newblock In Jin Akiyama, Mikio Kano, and Toshinori Sakai, editors, {\em
  Computational Geometry and Graphs}, pages 30--41, Berlin, Heidelberg, 2013.
  Springer Berlin Heidelberg.
\newblock \href {http://dx.doi.org/10.1007/978-3-642-45281-9_3}
  {\path{doi:10.1007/978-3-642-45281-9_3}}.

\bibitem{chrobak_karloff_a-lower-bound_1989}
M.~Chrobak and H.~Karloff.
\newblock A lower bound on the size of universal sets for planar graphs.
\newblock {\em SIGACT News}, 20, 1989.

\bibitem{chrobak_payne_a-linear-time-algorithm_1995}
M.~Chrobak and T.H. Payne.
\newblock A linear-time algorithm for drawing a planar graph on a grid.
\newblock {\em Information Processing Letters}, 54(4):241 -- 246, 1995.
\newblock URL:
  \url{http://www.sciencedirect.com/science/article/pii/002001909500020D},
  \href {http://dx.doi.org/https://doi.org/10.1016/0020-0190(95)00020-D}
  {\path{doi:https://doi.org/10.1016/0020-0190(95)00020-D}}.

\bibitem{defraysseix_pach_pollack_how-to-draw_1990}
H.~De~Fraysseix, J.~Pach, and R.~Pollack.
\newblock How to draw a planar graph on a grid.
\newblock {\em Combinatorica}, 10(1):41--51, Mar 1990.
\newblock URL: \url{https://doi.org/10.1007/BF02122694}, \href
  {http://dx.doi.org/10.1007/BF02122694} {\path{doi:10.1007/BF02122694}}.

\bibitem{fary_on-straight-line_1948}
Istv{\'a}n F{\'a}ry.
\newblock n straight line representation of planar graphs.
\newblock {\em Acta Sci.Math. Szeged}, 11:229--233, 1948.

\bibitem{frieze_on-the-length_1991}
Alan Frieze.
\newblock On the length of the longest monotone subsequence in a random
  permutation.
\newblock {\em Ann. Appl. Probab.}, 1(2):301--305, 05 1991.
\newblock URL: \url{https://doi.org/10.1214/aoap/1177005939}, \href
  {http://dx.doi.org/10.1214/aoap/1177005939}
  {\path{doi:10.1214/aoap/1177005939}}.

\bibitem{kurowski_a-1.234-louwer-bound_2004}
Maciej Kurowski.
\newblock A 1.235 lower bound on the number of points needed to draw all
  $n$-vertex planar graphs.
\newblock {\em Information Processing Letters}, 92:95--98, oct 2004.
\newblock \href {http://dx.doi.org/10.1016/j.ipl.2004.06.009}
  {\path{doi:10.1016/j.ipl.2004.06.009}}.

\bibitem{scheucher_etal_a-note-on-universal_2018}
Manfred Scheucher, Hendrik Schrezenmaier, and Raphael Steiner.
\newblock A note on universal point sets for planar graphs.
\newblock {\em ArXiv}, abs/1811.06482, 2018.

\bibitem{schnyder_embedding-planar-graphs_1990}
Walter Schnyder.
\newblock Embedding planar graphs on the grid.
\newblock In {\em Proceedings of the First Annual ACM-SIAM Symposium on
  Discrete Algorithms}, SODA '90, pages 138--148, Philadelphia, PA, USA, 1990.
  Society for Industrial and Applied Mathematics.
\newblock URL: \url{http://dl.acm.org/citation.cfm?id=320176.320191}.

\bibitem{Whitney}
Hassler Whitney.
\newblock Congruent graphs and the connectivity of graphs.
\newblock {\em Amer. J. Math.}, 54:150--168, 1932.

\end{thebibliography}

\end{document}